\documentclass[12pt,twoside]{article}
\usepackage[utf8]{inputenc}

\usepackage{amsmath,graphicx,microtype,amssymb,setspace,verbatim,multirow,url,bbm,rotating,booktabs,amsthm,subcaption}
\usepackage{libertine}
\usepackage{color} 
\usepackage{sgame}
\usepackage{fnpct}

\usepackage{xcolor}

\newif\ifshownotes
\shownotesfalse   

\makeatletter
\def\JL@red{red}
\NewCommandCopy{\JL@textcolor}{\textcolor}
\newcommand{\JL@tc}[3]{%
  \IfValueTF{#1}{\JL@textcolor[#1]{#2}{#3}}{\JL@textcolor{#2}{#3}}}
\RenewDocumentCommand{\textcolor}{o m m}{%
  \def\JL@c{#2}%
  \ifx\JL@c\JL@red
    \ifshownotes\JL@tc{#1}{#2}{#3}\fi
  \else
    \JL@tc{#1}{#2}{#3}%
  \fi}
\makeatother
\makeatletter
	\renewcommand{\abstract}[1]{\def \@abstract {#1}}
	\newcommand{\jelcodes}[1]{\def \@jelcodes {#1}}
	\newcommand{\keywords}[1]{\def \@keywords {#1}}
	\newcommand{\thanknotes}[1]{\def \@thanknotes {#1}}
	\newcommand{\contact}[1]{\def \@contact {#1}}
	\newcommand{\shortauthor}[1]{\def \@shortauthor {#1}}
	\newcommand{\shorttitle}[1]{\def \@shorttitle {#1}}
\makeatother

\jelcodes{}
\thanknotes{I thank Odilon C\^amara and Kevin He for helpful discussions. I used ChatGPT (5.6 Sol) for research assistance, which included exploratory proof development and simplification, algebraic and numerical checks, literature searches, and language editing. I used Claude Opus 5 for proofreading. All arguments and exposition were selected, developed, and verified by me; I take full responsibility for the final manuscript and any errors. }
\shortauthor{Jonathan Libgober}
\shorttitle{Disagreement}
\abstract{I introduce a robustness criterion for settings where players may disagree about the distribution over signals induced by endogenously chosen information. In the spirit of \cite{hansensargent}, an information choice is evaluated by its worst-case payoff over beliefs within an $\eta$-entropy ball around the signal distribution it induces. I apply this formulation to buyer-optimal learning in bilateral trade \citep{RS}. If trade is always efficient, disagreement raises the price and twists the buyer's demand curve while preserving full trade. When trade may be inefficient, the same characterization holds conditional on trade with an endogenously amplified disagreement level.}

\usepackage{enumitem}
\usepackage[paperwidth=8.5in, paperheight=11in]{geometry}
\usepackage{setspace}
\usepackage[hang,flushmargin]{footmisc} 
\usepackage[authoryear]{natbib}

\makeatletter
\usepackage{fancyhdr}
\makeatother

\newcommand\blfootnote[1]{%
  \begingroup
  \renewcommand\thefootnote{}\footnote{#1}%
  \addtocounter{footnote}{-1}%
  \endgroup
}

\usepackage{authblk}
\makeatletter
\def \maketitle { 
	\thispagestyle{empty}
	\vspace*{0.1in}
	\blfootnote{Contact: libgober.economics@gmail.com.  \@thanknotes}
	\begin{center}
	\begin{minipage}{5.2in}
	\begin{center}
	{\large {\textbf{\@title}}}
	
	\vspace{0.2in}
	
	{\textsc{\@author}}
	
	\vspace{0.2in}
	
	{\@date}
	\end{center}
	
	\ifx\@abstract\@empty
	\relax
	\else
	{\small{\textsc{Abstract.} \@abstract}}
	\fi
	
	\ifx\@keywords\@empty
	\relax
	\else
	\vspace{0.2in}
	
	{\small\textsc{Keywords.} \@keywords.}
	\fi
	
	\ifx\@jelcodes\@empty
	\relax
	\else
	{\small\textsc{JEL Codes.} \@jelcodes.}
	\fi
	
	\end{minipage}
	\end{center} }
\makeatother

\makeatletter
\def\@seccntformat#1{\csname the#1\endcsname.\ }
\makeatother

\usepackage{sectsty}

\allsectionsfont{\noindent\normalsize}
\sectionfont{\centering\normalsize}
\paragraphfont{\textnormal}

\DeclareMathOperator*{\argmax}{arg\,max}
\DeclareMathOperator*{\argmin}{arg\,min}

\providecommand{\U}[1]{\protect\rule{.1in}{.1in}}

\newtheorem*{theorem}{Theorem}

\newtheorem*{claim}{Claim}

\newtheorem*{corollary}{Corollary}

\newtheorem{lemma}{Lemma}

\newtheorem*{proposition}{Proposition}

\newtheorem*{definition}{Definition}

\theoremstyle{definition}
\newtheorem{example}{Example}

\usepackage{overpic}
\usepackage{hyperref}

\usepackage{tabu}

\makeatletter
\renewcommand\@biblabel[1]{}
\makeatother

\title{Disagreeing About What the Buyer Might Learn}
\author{Jonathan Libgober \\ University of Southern California}
\date{September 24, 2026}

\begin{document}
\maketitle
\newpage

\section{Introduction}
This paper considers the problem a buyer faces when deciding what to learn about her value for a product, anticipating that the seller may hold a different belief about the demand curve induced by this information. I have in mind that the interaction is one where the buyer decides freely which information sources to consult privately when learning about her willingness to pay. While the seller knows what the buyer chose, he might nevertheless disagree about which signals are likely to be generated; for instance, he may believe that some favorable signals are more likely than the buyer expects. I study how the prospect of belief disagreement influences the choice of information itself.

I adopt a robustness approach in the spirit of \cite{hansensargent} in order to discipline belief disagreement in a way that generates sharp predictions. \cite{morris1995common} identifies the need for such discipline as one
of the strongest arguments for the common prior assumption, despite the apparent ubiquity of ``agreeing to disagree'' that this assumption rules out. This concern is especially salient when information \emph{itself} is endogenous, as in my setting. I posit that, while the buyer expects that her posterior expected value will follow a distribution $G$, the seller may perceive a different distribution, say $H$, when selecting the price. Anticipating this, the buyer chooses $G$ to maximize her worst-case expected payoff over seller beliefs in the set:
\begin{equation*}
\{H : D_{KL}(H||G) \leq \eta\},
\end{equation*}

\noindent where $D_{KL}$ denotes the Kullback-Leibler divergence. This criterion seeks a favorable payoff guarantee no matter which admissible seller belief will arise. The parameter $\eta$ dictates how much beliefs might differ, while $G$ influences which kinds of disagreements are feasible. At one extreme, the seller cannot think the buyer learned something if the buyer learns nothing. 

When $\eta=0$, my model coincides with that of \cite{RS}. In this model, the buyer can costlessly choose any Blackwell experiment to privately learn about her willingness-to-pay before the seller selects a price. When gains from trade are always positive, buyer-optimal learning induces trade with probability one. Hence, information influences the surplus division through price, without distorting the allocation. This favorable price is sustained by an equal-revenue distribution making the seller indifferent among all prices in its support. 

The \cite{RS} benchmark provides an incisive starting point for the analysis of belief disagreement. The reason is that the buyer benefits from choosing information by influencing the \emph{seller's perception} of demand. Disagreement therefore concerns the very distribution the buyer chooses to influence pricing. Since exact seller indifferences support the equilibrium price, an arbitrarily small disagreement that places relatively more probability on high signals can make a higher price strictly optimal. Such changes can eliminate buyer surplus entirely in the Roesler-Szentes information structure. The central issue is how the buyer should redesign information to protect herself when that choice determines both the demand she induces and the set of possible seller perceptions of that demand. 

When gains from trade are always positive, disagreement influences information choice as follows. First, the price the buyer faces is higher. Second, the distribution of expected willingness-to-pay is ``twisted'' in a way that depends on disagreement. This twisting arises from an indifference condition for the seller between the target price and any higher price. The distinction with \cite{RS} emerges because this indifference might be enforced by a seller belief that depends on the deviation. Typically, no single $H$ will make the seller indifferent among all deviations. For each deviation, the KL constraint determines which seller belief makes it most attractive. In the equilibrium outcome, the lowest price in the support of the buyer's distribution is selected for any admissible $H$, including if $H=G$. Full trade emerges in equilibrium, just as in the no-disagreement case. 

Despite the main message being that disagreement essentially leads to a twisting of the no-disagreement solution, the proof strategy is rather different from \cite{RS}. The idea in \cite{RS} is to replace an arbitrary buyer choice with one that induces full trade and the same seller profit. While the \emph{intuition} from \cite{RS} can motivate some parts of the analysis, fundamentally it may be impossible to replace an arbitrary outcome with another that induces full trade and the same seller profit (I present an example in the Supplemental Appendix). Although disagreement can generate outcomes with lower seller profit, my theorem shows that moving away from the full-trade optimum cannot convert this into higher buyer surplus. My strategy is to instead consider an auxiliary problem where the buyer chooses $G$ to minimize the maximum profit the seller can achieve over any $H$. I then show that the original model and the auxiliary model have the same solution by showing separately that no other $G$ can raise buyer surplus further. 

If trade is potentially inefficient, then the buyer uses a similar information structure which possibly introduces a single signal where trade does not occur. Outside of that signal, information resembles the efficient-trade solution but with \emph{exacerbated disagreement}. The intuition is that Nature can lower the perceived probability of trade whenever trade occurs with probability less than one, to sustain more disagreement conditional on trade. I characterize this amplification exactly; in the worst-case, the buyer and seller ``agree to disagree'' about the probability of trade.

\subsection{Related Literature}

This paper applies a robustness criterion, motivated by prospective belief disagreement, to an information-design problem. Perhaps the closest precedent comes from work on persuasion studying the implications of belief disagreement over \emph{initial priors}. \cite{alonsocamara} study persuasion with heterogeneous priors with a commonly understood experiment. \cite{kosterina2022} considers a similar problem but assuming worst-case priors; while the feasible set of receiver priors is fixed, the worst-case belief may still vary with information. My worst-case criterion concerns disagreement over the \emph{signal distribution} induced by an information structure. In this sense the exercise is related to \cite{dworczakpavan2022}, studying a model where the sender is uncertain about the receiver's additional information rather than belief per se. 

In the pricing setting, work studying robustness has largely focused on seller-side  uncertainty \citep{du2018,brooksdu2021}; \cite{debroesler} relate buyer-optimal and seller-pessimal information in the context of multidimensional screening. Meanwhile, \cite{langwasser} study a \emph{buyer} with $\alpha$-maxmin preferences who may face an ambiguous information structure.  My objective differs since the worst-case is with respect to disagreement regarding likely signal realizations. My buyer is Bayesian over her own information. 

This paper also joins a literature asking which features of buyer-optimal learning change in environments beyond \cite{RS}. \cite{RavidRS} studies costly information acquisition unobserved by the seller, while \cite{condorelliszentes} allows the buyer to choose the value distribution itself. Other work studies initial private information \citep{chopraely} or subsequent seller disclosure \citep{TerstiegeWasser}. This setting is a natural laboratory to study the impact of worst-case belief disagreement. Seller indifference is replaced by distribution-specific indifference, with the entropy constraint directly influencing how the information choice is transformed.

\section{Model} \textcolor{red}{Supposed to allow the buyer and seller to agree to disagree, and this is maybe not as clearly written since the seller observes $\pi$; this should be something like the seller sees what the buyer thinks but ultimately thinks it's something else. Here it doesn't really matter what the actual information is since the seller could always believe that the buyer got no information and the signal distribution was constant, so there is this form of indeterminacy.}

\noindent A seller with production cost $c  \in [0,1]$ seeks to sell one unit of a good to a buyer with value $v \sim F$, where $F$ is a continuous distribution with support $[0,1]$ and mean $\mu$. The buyer's information about $v$ is determined by an information structure  $\mathcal{I} : [0,1] \rightarrow \Delta(S)$ which she chooses and whose realization $s \in S$ is privately observed. After the buyer chooses an information structure, the seller chooses a price to maximize his perceived expected profit. I normalize signal realizations so that:\footnote{This is without loss; if multiple signals can be generated with the same posterior expectation, Nature can induce the same anticipated demand curve compared to the case where signals are collapsed by matching the distribution over signals with equal posterior expectations as specified by $\mathcal{I}$.}  

\begin{equation*} 
s= \mathbb{E}_{v \sim F, s \sim \mathcal{I}}[v \mid s]
\end{equation*}

\noindent Let $G$ denote the distribution over posterior expectations $s$ induced by $\mathcal{I}$; throughout the paper I suppress the implicit dependence on $\mathcal{I}$. A distribution over posterior expectations is \emph{Bayes plausible} with respect to the prior $F$ if it can be induced by some information structure, which holds if and only if for all $x$, 
\begin{equation*} 
\int_{0}^{x} G(s)ds \leq \int_{0}^{x} F(s) ds, 
\end{equation*}
with $\int_{0}^{1} sdG(s)=\mu$. 

\cite{RS} study this model when the buyer's information structure is \emph{common knowledge}. By contrast, I consider a buyer concerned that even though the seller observes the buyer's choice, he might nevertheless ``agree to disagree'' about the distribution over $s$ induced by that choice. Let $D_{KL}(f || g)$ denote the Kullback-Leibler divergence between distributions $f$ and $g$.\footnote{That is, $D_{KL}(H||G) = \int \log(dH/dG)dH$ when $H$ is absolutely continuous with respect to $G$ and $\infty$ otherwise.} If the seller anticipates the distribution over the buyer's signals is $H$, then he selects an optimal price according to $H$, breaking ties in favor of the buyer, and assuming the buyer will buy if indifferent:\footnote{I discuss the role of tiebreaking in Section \ref{sec:extensions}.} 

\begin{equation*} 
p^{*}(H) = \min \argmax_{p \in [c,1]} (p-c) \mathbb{P}_{H}[s \geq p],
\end{equation*}

\noindent The buyer then selects $G$ to maximize payoffs against the worst-case $H$ whose KL-divergence with $G$ is no more than $\eta$: 

\begin{equation*} 
\inf_{H : D_{KL}(H||G) \leq \eta} ~~~\int_{p^{*}(H)}^{1} (s- p^{*}(H))dG(s). 
\end{equation*}

\noindent The \cite{RS} setting emerges when $\eta=0$. I emphasize that $H$ is a distribution over signals labeled by the buyer's posterior expectation, and \emph{not} necessarily the seller's. Accordingly, $H$ need not be Bayes plausible with respect to $F$ even though $G$ must be.

\subsection{Discussion} \label{sec:nodisag}
When $\eta=0$, there is no disagreement about the signal distribution. When $c=0$, \cite{RS} showed the buyer optimally chooses: 

\begin{equation} 
G(s) = \begin{cases} 0 & 0 \leq s <q \\ 1- \frac{q}{s}  & q \leq s < B \\ 1 & B \leq s \leq 1\end{cases}.\label{eq:RSInfo}
\end{equation}

\noindent This distribution is continuous at $q$ but has an atom at $s=B$. The seller is indifferent between all prices in the support of $G$, and charges a price equal to the minimum in the support. \cite{RS} use a clever replacement argument to prove this: Any outcome achieving some buyer payoff-seller profit pair can be replaced by one of the form (\ref{eq:RSInfo}) yielding full trade and equal seller profit. Thus, any efficiency gains go to the buyer. However, the distribution $H_{\varepsilon}=(1-\varepsilon)G+\varepsilon \delta_{B}$ (where $\delta_{x}$ is a point mass at $x$) has $D_{KL}(H_{\varepsilon}||G) < \eta$ for $\varepsilon$ sufficiently small, and for every $\varepsilon > 0$ the seller \emph{strictly} prefers $p=B$---yielding zero buyer surplus. Still, not all choices are \emph{so} sensitive to disagreement; if an interior price is uniquely strictly optimal, it may be impossible to induce a high price with only a small change in the perceived signal distribution.

\section{Main Result and Proof Idea} 

To understand the form of the buyer's information structure choice, I first consider how the buyer can ensure a full-trade outcome, say at a price of $r$. Since the seller should find it weakly optimal to set a price of $r$, deterring a deviation to $p > r$ requires that $p \cdot \mathbb{P}_{H}[s \geq p] \leq r$ for any admissible $H$. But since the buyer chooses $G$ and not $H$, the question is how to translate this restriction on $H$ to a restriction on $G$. To do this, define:

\begin{equation*}
d_{KL}(x||y)=x\log\left(\frac{x}{y}\right)
+(1-x)\log\left(\frac{1-x}{1-y}\right)
\end{equation*}

\noindent to be the KL-divergence for two Bernoulli distributions assigning probabilities $x$ and $y$, respectively, to the same event. 

\begin{definition}\label{defn}
The \emph{$\tilde{\eta}$-disagreement inversion at $x\in (0,1)$} is the value $I_{\tilde{\eta}}(x) \in [0,1]$ which maps $x$ into the unique value weakly below $x$ satisfying:

\begin{equation*}
d_{KL}(x||I_{\tilde{\eta}}(x))=\tilde{\eta}.
\end{equation*}

\noindent I use the convention that $I_{\tilde{\eta}}(0)=0$ and $I_{\tilde{\eta}}(1)=e^{-\tilde{\eta}}$. 

\end{definition} 

\noindent As the proof below shows, $I_{\eta}(x)$ is the pivotal object because if $\mathbb{P}_{G}[s \geq p]= I_{\eta}(x)$ then $x$ is the largest probability an admissible $H$ can assign to the event $\{s \geq p\}$. Thus, ensuring the seller sets $r$ instead of $p > r$ requires $\mathbb{P}_{G}[s \geq p] \leq I_{\eta}(r/p)$. In fact, this inequality is tight: 

\begin{theorem}[$c=0$]\label{thm}
Buyer-optimal information induces trade with probability one and, for some $r^{*} \leq B^{*}$, can be chosen so that: 

\begin{equation*} 
\mathbb{P}_{G}[s \geq p] = I_{\eta}(r^{*}/p),  ~~~~ r^{*} < p \leq B^{*}, 
\end{equation*}

\noindent with support $[r^{*}, B^{*}]$. The seller charges $r^{*}$ under every admissible belief. 
\end{theorem}

\noindent The theorem gives a direct mapping between the disagreement and no-disagreement constructions. Fixing the profit $r^{*}$, the disagreement solution applies a disagreement inversion to the tail probabilities for the equal-revenue distribution; \cite{RS} shows that such distributions can induce full trade while fixing seller profit. This construction introduces an atom of size $1-e^{-\eta}$ at the bottom of the support. Moreover, while each higher price can be made indifferent to the buyer's intended price under some admissible seller belief, the $H$ doing so will generally differ across prices.   Figure \ref{fig:InfoSolution} illustrates the resulting buyer beliefs and several such seller beliefs. In general, no single $H$ generates all of these price indifferences.

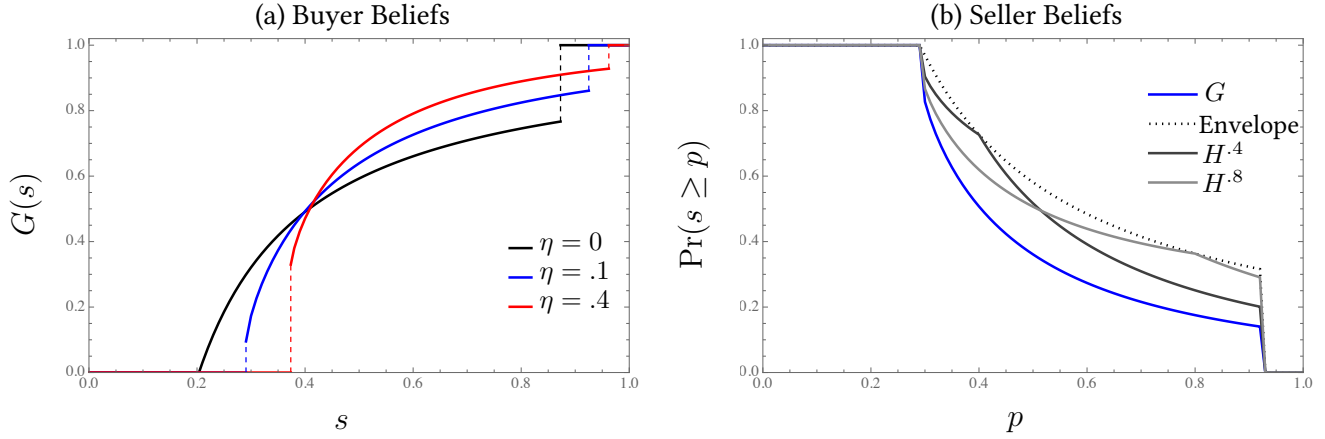
\begin{figure}[t]
    \centering

    \begin{subfigure}[t]{.46\textwidth}
        \centering
        \caption{Buyer Beliefs}
        \begin{overpic}[width=\linewidth]{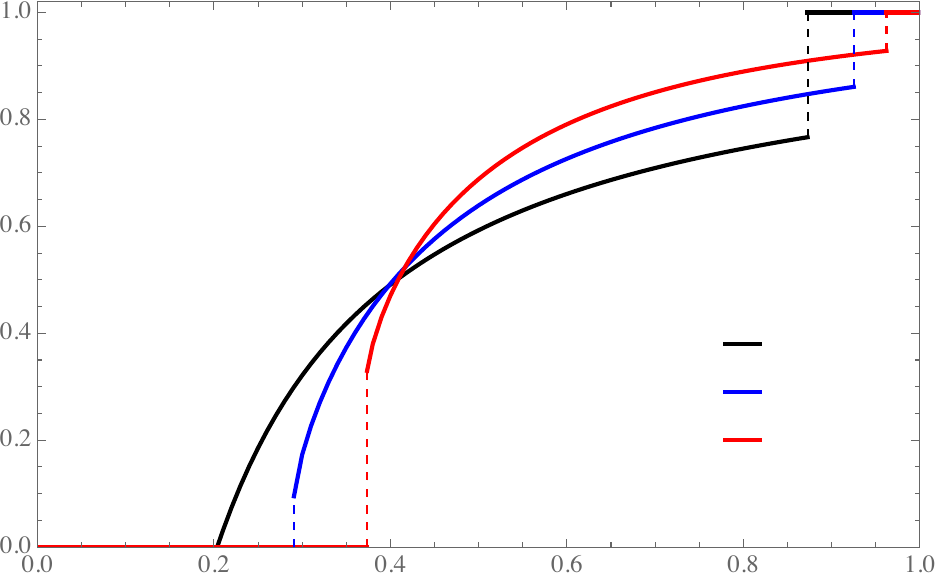}
            \put(48,-5){\makebox(0,0){$s$}}
            \put(-7,33){\rotatebox{90}{\makebox(0,0){$G(s)$}}}

            \put(82.5,26){\makebox(0,0)[l]{\footnotesize $\eta=0$}}
            \put(82.5,21){\makebox(0,0)[l]{\footnotesize $\eta=.1$}}
            \put(82.5,16){\makebox(0,0)[l]{\footnotesize $\eta=.4$}}
        \end{overpic}
        \label{fig:buyerdist}
    \end{subfigure}
    \hfill
    \begin{subfigure}[t]{.46\textwidth}
        \centering
        \caption{Seller Beliefs}
        \begin{overpic}[width=\linewidth]{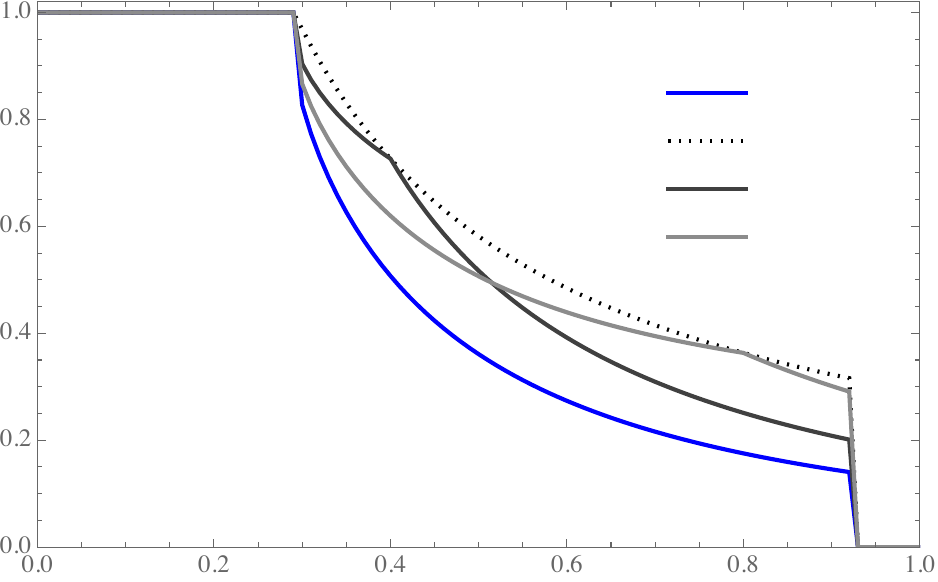}
            \put(48,-5){\makebox(0,0){$p$}}
            \put(-8,33){\rotatebox{90}{\makebox(0,0){$\Pr(s\geq p)$}}}

            \put(81,52.5){\makebox(0,0)[l]{\footnotesize $G$}}
            \put(80,46.5){\makebox(0,0)[l]{\footnotesize Envelope}}
            \put(80.5,42){\makebox(0,0)[l]{\footnotesize $H^{.4}$}}
            \put(80.5,37){\makebox(0,0)[l]{\footnotesize $H^{.8}$}}
        \end{overpic}
        \label{fig:sellerdist}
    \end{subfigure}

    \caption{\footnotesize
    Illustration of the solution for $v\sim U[0,1]$ and $c=0$.
    Panel (a) plots the buyer-optimal distribution of posterior means
    for $\eta=0,.1,.4$. For $\eta>0$, the distribution has an atom of
    size $1-e^{-\eta}$ at the bottom of its support.
    Panel (b) sets $\eta=.1$. The solid blue curve is the buyer-optimal $G$, while the dotted curve is the pointwise upper envelope of
    seller-perceived demand over beliefs $H$ satisfying
    $D_{\mathrm{KL}}(H\Vert G)\leq\eta$. The two gray curves are
    representative beliefs, $H^{.4}$ and $H^{.8}$, attaining the envelope at prices $.4$ and $.8$, respectively; each belief makes the seller indifferent between that price and the lowest price in  the support. 
    Each point on the envelope corresponds to a seller belief tailored to that price; the same belief need not generate the envelope at other prices. 
    }
    \label{fig:InfoSolution}
\end{figure}

Now, while the price-by-price derivation illustrates how disagreement transforms information choice, more work is needed to prove the theorem. First, the price-by-price argument only gives upper bounds on the tail probabilities needed to ensure that no price above $r$ yields profit more than $r$. It leaves open how low $r$ can be, and whether these bounds can be tight for a single $G$. Second, the argument takes full trade as given; it leaves open whether the buyer might benefit from sacrificing trade, say, to avoid high prices. The formal proof addresses these two issues. As mentioned, the replacement strategy of \cite{RS} cannot resolve the second issue because it \emph{may not be possible} to start with an arbitrary feasible buyer surplus-seller profit pair and induce full trade while fixing seller profit; see Example \ref{ex:nocanon} in the Supplemental Appendix.

\subsection{Proving the $c=0$ Theorem}\label{sec:c0}

My proof strategy is to solve an auxiliary problem where the buyer chooses an information structure to minimize the seller's maximum profit across all possible $H$, showing that this solution involves trade with probability 1. I then show separately that the solution to this auxiliary problem is also buyer-optimal because the buyer cannot do better by introducing inefficiency into trade. \\

\subsubsection{The Auxiliary Problem's Solution}

Anticipating the $c > 0$ case, the following argument takes Supp$F=[\underline{v},\overline{v}]$ where $\underline{v} \leq 0$.

\medskip 

\noindent \textbf{Step One: Characterize, Price-by-Price, the $H$ that Maximizes Seller Revenue} Given a price $p$, the revenue-maximizing belief $H$ solves: 

\begin{equation*}
 \max_{H} H(\{s \geq p\}) \text{ such that } D_{KL}(H || G) \leq \eta.
\end{equation*}

\noindent The key behind Step One is the \emph{KL chain rule} \citep{CoverThomas2006}, which implies: 

\begin{multline*}
D_{KL}(H || G) = d_{KL}(\mathbb{P}_{H}[s \geq p] ||  \mathbb{P}_{G}[s \geq p]) \\ +\mathbb{P}_{H}[s \geq p]D_{KL}(H || G; s \geq p)+\mathbb{P}_{H}[s < p]D_{KL}(H || G; s < p); \end{multline*}

\noindent Here $D_{KL}(H||G;s\geq p)$ denotes $D_{KL}(H(\cdot\mid s\geq p)||G(\cdot\mid s\geq p))$,
and similarly for $s<p$. Since the conditional terms are nonnegative,

\begin{equation*}
D_{KL}(H || G) \geq d_{KL}(\mathbb{P}_{H}[s \geq p] ||  \mathbb{P}_{G}[s \geq p]).
\end{equation*}

\noindent Furthermore, for any feasible probability of purchase, this bound is attained by reweighting $G$ above and below $p$, fixing the conditional distributions. Since seller profit at $p$ depends only on the probability of purchase, I rewrite the problem as:

\begin{equation*}
 \max_{H} H(\{s \geq p\}) \text{ such that }  d_{KL}(\mathbb{P}_{H}[s \geq p] ||  \mathbb{P}_{G}[s \geq p]) \leq \eta.
\end{equation*}

\noindent \textbf{Step Two: Minimize the Seller's Maximum Profit Across all Possible $H$} Next, I propose a candidate $G$ such that the seller obtains no more than some profit target $r$, no matter what $H$ is. I use the following observation: 

\begin{lemma}[Bounding Seller Deviation Payoffs] \label{lem:bound}
The seller's profit at price $p$ is at most $r < p$ for every admissible $H$ if and only if \begin{equation*} \mathbb{P}_{G}[s \geq p]\leq I_{\eta}(r/p). \end{equation*}

\end{lemma}

This Lemma helps identify a bound on what $r$ can be. This step uses the following inequality, which holds whenever $r \leq x$: 
\begin{equation}
s \leq r+ \max\{s-x,0\}+\int_{r}^{x}\mathbf{1}[s \geq p]dp. \label{eq:prelimineq}
\end{equation}
At $s<r \leq x$ this reduces to $s < r$, while equality holds whenever $r \leq s < x$ or $s  \geq x$. Now, take the expectation with respect to $s$ using $G$: 

\begin{equation*}
\mu \leq r+ \mathbb{E}_{G}\left[\max\{s-x,0\}\right]+\int_{r}^{x}\mathbb{P}_{G}[s \geq p]dp.
\end{equation*}
Using Lemma \ref{lem:bound}, the identity $\mu-x+\int_{\underline{v}}^{x} G(s)ds=\mathbb{E}_{G}\left[\max\{s-x,0\}\right]$, and the fact that $G$ is a mean preserving contraction of $F$:
\begin{equation}
x \leq \overbrace{r+ \int_{\underline{v}}^{x} F(s)ds+\int_{r}^{x}I_{\eta}(r/p)dp}^{:=D_{\eta}(r,x)}.\label{mainsuffeq}
\end{equation}

\noindent But (\ref{mainsuffeq}) is a condition on $F$ and $\eta$ alone. So, if the seller can obtain no more than $r$ for any $H$, then $0 \leq \min_{x \in [r, \bar{v}]} D_{\eta}(r,x)-x$ must hold. Suppose $r^{*}$ solves: 

\begin{equation*} 
\min\{ r \in [0,\mu] : \min_{x \in [r, \bar{v}]} D_{\eta}(r,x)-x \geq 0\}.
\end{equation*}

\noindent In particular, $\min_{x \geq \mu} D_{\eta}(\mu,x)-x \geq 0$, so the feasible set is nonempty; continuity of $D_{\eta}(r,x)$ implies that this minimum is attained. Furthermore, $r^{*} \leq \mu$ is immediate since there is a candidate---no information---at which the seller gets $\mu$; and $D_{\eta}(0, \bar{v}) - \bar{v}=\int_{\underline{v}}^{v} F(s)ds - \bar{v}=-\mu <0$, so that $r^{*} > 0$ by continuity\footnote{Continuity at $r=0$ holds since $I_{\eta}(y) \leq y$ implies $0 \leq \int_{r}^{x} I_{\eta}(r/p) dp \leq r \int_{r}^{x} \frac{dp}{p} = r \log(x/r)$, which converges to 0 as $r \rightarrow 0$.} of $D_{\eta}(r,x)-x$.  

The necessary condition for a seller to be held to no more than $r$ is also sufficient. Set:

\begin{equation*}
G_{r}(p)= \begin{cases} 0 & p < r \\ 1-I_{\eta}(r/p) & r \leq p < B \\ 1 & p \geq B  \end{cases}
\end{equation*}

\noindent When $\eta=0$, $G_{r}$ reduces to the solution described in section \ref{sec:nodisag}; thus, the impact of disagreement is to replace the tail probabilities with their disagreement inversions.\footnote{As mentioned, $G_{r}(p)$ has an atom at $r$ since $I_{\eta}(1)=e^{-\eta}$.} 

\begin{lemma}  \label{lem:BP}
For any $r \in (0, \mu]$ with $\min_{x \in [r, \bar{v}]} D_{\eta}(r,x)-x \geq 0$, there exists $B$ such that $G_{r}(p)$ is Bayes plausible and ensures the seller's maximal profit over all admissible $H$ is $r$. 
\end{lemma}

\noindent The proof checks conditions for Bayes plausibility and the condition from Lemma \ref{lem:bound}.

\subsubsection{Buyer-Optimal Information}

Stepping back, $G_{r^{*}}$ ensures that no $H$ can give the seller (subjective) profit above $r^{*}$, and this is the lowest the seller can be held to if $H$ maximizes seller profit. But $H$ minimizing buyer surplus \emph{need not} imply that it maximizes seller profit. I now show that no other $G$ can increase buyer surplus, say, by inducing less trade in an attempt to lower the price.

Here I restrict attention to information structures with nonnegative posterior means, which holds automatically when $c=0$. For the buyer to be better off under an alternative $G$, there would need to be some $\tilde{p}$ such that:

\begin{equation*}
\int_{\tilde{p}}^{\bar{v}} (s-\tilde{p}) dG(s) \geq \mu - r^{*}
\end{equation*}

\noindent Let $p^{\dagger}$ denote the value of $\tilde{p}$ at which equality holds, which exists since the left hand side is $\mu$ at $\tilde{p}=0$, 0 at $\tilde{p}=\bar{v}$, and continuous. Then $\int_{0}^{p^{\dagger}} G(s)ds = p^{\dagger}- r^{*}$, which implies $p^{\dagger} > r^{*}$; otherwise, $G$ has no mass below $r^{*}$, so neither can any admissible $H$ (since $H$ must be absolutely continuous with respect to $G$). Thus, $r^{*}$ induces trade with probability 1 for any $H$, meaning the buyer cannot improve on $\mu-r^{*}$. 

The fundamental claim is the following: 

\begin{claim}[Price inducibility]
A candidate $G$ can only be an improvement if for every $p > p^{\dagger}$: 

\begin{equation}
\mathbb{P}_{G}[s \geq p] < I_{\eta}\left(\frac{r^{*}}{p-p^{\dagger}+r^{*}}\right). \label{eq:PIEq}
\end{equation}

\end{claim}

\noindent The logic is that if this inequality fails at some $p$, then the work in Step One implies Nature can choose $H$ so that the seller earns $\frac{r^{*}p}{p-p^{\dagger}+r^{*}} > r^{*}$ at $p$ by rescaling the distribution below $p$. The definition of $p^{\dagger}$ implies that under this rescaling, every price $q \leq p^{\dagger}$ yields lower profit, so every seller best response weakly exceeds $p^{\dagger}$. So any candidate improvement must satisfy the stated inequality. Section \ref{sec:proofs} presents the formal argument. 

I prove the theorem by showing (\ref{eq:PIEq}) cannot hold.  By minimality and continuity, the constraint $D_{\eta}(r,x)-x \geq 0$ must bind at $r^{*}$, with the minimum attained at some $x^{*} \in [r^{*},\overline{v}]$. In particular, $D_{\eta}(r^{*},x^{*})-x^{*}=0$. Since this expression becomes $\int_{\underline{v}}^{r^{*}} F(s)ds=0$ when $x^{*}=r^{*}$, Bayes plausibility would require $G$ has no mass below $r^{*}$ either, so that $r^{*}$ induces trade with probability 1 under every admissible $H$; since no lower price can be optimal, the buyer cannot obtain more than $\mu-r^{*}$. So $x^{*} > r^{*}$ for any candidate improvement, and $D_{\eta}(r^{*},x^{*})-x^{*}=0$ yields:
\begin{equation} 
 \int_{\underline{v}}^{x^{*}} F(s)ds -(x^{*}-r^{*})= -\int_{r^{*}}^{x^{*}} I_{\eta}(r^{*}/p) dp.\label{eq:ContraEQ}
\end{equation}

\noindent Bayes plausibility then gives $\int_{\underline{v}}^{p^{\dagger}}F(s)ds - (p^{\dagger}-r^{*}) \geq 0$. Since $ \int_{\underline{v}}^{x} F(s)ds -(x-r^{*})$ is decreasing in $x$ and (\ref{eq:ContraEQ}) implies changing $p^{\dagger}$ to $x^{*}$ yields something negative, I infer that $x^{*} > p^{\dagger}$. Price inducibility implies: 

\begin{align*} 
\int_{\underline{v}}^{x^{*}} G(s)ds &= p^{\dagger}- r^{*} + \int_{p^{\dagger}}^{x^{*}}G(s)ds  > p^{\dagger}- r^{*} + \int_{p^{\dagger}}^{x^{*}}\left(1- I_{\eta} \left( \frac{r^{*}}{s-p^{\dagger}+r^{*}} \right) \right)ds \\ & = x^{*} -r^{*} - \int_{r^{*}}^{x^{*}-p^{\dagger}+r^{*}} I_{\eta}(r^{*}/u) du > x^{*} - r^{*} - \int_{r^{*}}^{x^{*}} I_{\eta}(r^{*}/u) du. 
\end{align*}

\noindent where the last equality follows by change of variables and algebra and the last inequality holds since $p^{\dagger} > r^{*}$. But (\ref{eq:ContraEQ}) then yields: 

\begin{equation*} 
\int_{\underline{v}}^{x^{*}} G(s)ds > \int_{\underline{v}}^{x^{*}} F(s) ds, 
\end{equation*}

\noindent contradicting Bayes plausibility.

\subsection{Generalizing to $c > 0$}

When $c > 0$, seller profit minimization need not correspond to buyer optimal information since trade could be inefficient. Thus, Steps One and Two need not yield buyer-optimal information. This creates a sharper complication extending the results to $c> 0$ relative to \cite{RS}, where a modified version of their replacement argument delivers the generalization. Here, a separate argument is needed to recover the structure of buyer-optimal information.

The key step is to condition on the event that the buyer's signal is above $c$; let $\alpha$ denote the probability of this event. This is an endogenous outcome that the buyer maximizes over to determine her optimal information structure. Introducing $\alpha$ lets me reframe the problem conditional on $s > c$ as analogous to the $c=0$ case with exacerbated disagreement. Define: 

\begin{equation*}
\eta^{*}(\alpha)=-\log\left(
\frac{e^{-\eta}-1+\alpha}{\alpha}\right),
\end{equation*} 

\noindent for $\alpha > 1-e^{-\eta}$; otherwise, let $\eta^{*}(\alpha)=\infty$

\begin{lemma}\label{lem:conditioning}
If $\alpha \leq 1-e^{-\eta}$, then buyer surplus is zero. If $\alpha>1-e^{-\eta}$, then fixing the distribution of $v$ conditional on $s > c$, the buyer's problem conditional on this event is equivalent, after shifting values and prices by $c$, to the $c=0$ problem with disagreement parameter $\eta^*(\alpha)$.  Unconditional buyer surplus and seller profit evaluated under the buyer's beliefs are $\alpha$ times their corresponding conditional values.
\end{lemma}

In light of this lemma, I now describe a replacement argument showing that the resulting buyer-optimal information structure takes a similar form as the $c=0$ case. First, start with an arbitrary information structure $G$, and let $\alpha_{G} = \mathbb{P}_{\tilde{s} \sim G}[\tilde{s} > c]$. Consider a replacement of the following form:

\begin{itemize} 
\item Draw an initial signal, $\tilde{s}$, according to $G$, with the buyer only observing whether $\tilde{s} >c$. 

\item If $\tilde{s} > c$, draw $s$ so that $s-c \sim \tilde{G}$, where: 

\begin{equation*} 
\tilde{G}(t)= \begin{cases} 0 & t < r_{G} \\ 1-I_{\eta^{*}(\alpha_{G})}(r_{G}/t) & r_{G} \leq t < B \\ 1 & t \geq B  \end{cases},
\end{equation*}

\noindent for some $r_{G}$. 

\item When $\tilde{s} \leq c$, draw signal $s_{*}=\mathbb{E}[v \mid \tilde{s} \leq c] \leq c$.
\end{itemize}

\noindent In particular, $r_{G}$ can be set equal to the value for $r^{*}$ from Section \ref{sec:c0} for the distribution of $v-c$ conditional on $\tilde{s} >c$, with disagreement parameter $\eta^{*}(\alpha_{G})$. Both replacements are Bayes plausible conditional on the respective events, so combining them remains Bayes plausible with respect to $F$. And by construction, the buyer must do weakly better conditional on $\tilde{s} > c$ with this replacement information structure than the original. Under the replacement, trade occurs with probability 1 whenever $\tilde{s} > c$ and does not occur when $\tilde{s} \leq c$. Hence $\alpha_{G}$ is also the probability of trade.

This argument shows it is without loss to restrict to information structures that are derived from the $c=0$ solution, but it leaves open how to find the corresponding values for $\alpha$ and $r$. Again carrying through the modifications in Section \ref{sec:proofs}, I derive the following as an upper bound on total surplus:\footnote{To avoid domain issues in defining $\Gamma_{\eta}(\alpha,r)$, assume this is 0 if $\alpha \leq 1-e^{-\eta}$.} 

\begin{equation*} 
\Gamma_{\eta}(\alpha,r):=\min_{x \in [-c,1-c]}  \alpha \min\{x,r\}+\mu - c- x + \int_{0}^{c+x}F(s)ds+\alpha \int_{r}^{\max\{x,r\}} I_{\eta^{*}(\alpha)}(r/t)dt.
\end{equation*}

\noindent 

Note that $\Gamma_{\eta}(\alpha,r) \geq \alpha r$ must hold, since $\alpha r$ is the surplus the seller obtains (from the buyer's perspective). Also, $\Gamma_{\eta}(\alpha,r) \geq \mu-c$ must hold since the replacement forgoes trade only when $s \leq c$, in which case trade generates weakly negative surplus; hence, total surplus must be weakly greater than the always-trade benchmark. Section \ref{sec:proofs} shows these conditions ensure a Bayes plausible information structure can be constructed, yielding the following theorem: 

\begin{theorem}[$c>0$]\label{thm}
The buyer's maximal robust surplus when $c>0$ is strictly positive if and only if $c < \bar{c}$, where for $\eta >0$:\footnote{When $\eta=0$, $\bar{c}=1$ as per \cite{RS}.}

\begin{equation*} 
\bar{c} = \mathbb{E}[v \mid v \geq F^{-1}(e^{-\eta})]
\end{equation*}

\noindent For any $c < \bar{c}$, the buyer-optimal probability that $s > c$, $\alpha^{*}$, and seller profit conditional on $s > c$, $r^{*}$, maximize:

\begin{equation*} 
\max_{\Gamma_{\eta}(\alpha,r) \geq \max\{\mu-c, \alpha r\}} [\Gamma_{\eta}(\alpha,r) - \alpha r].
\end{equation*}

\noindent Buyer-optimal information can be taken to involve a single no-trade signal with probability $1-\alpha^{*}$, and otherwise a signal supported on $[c+r^{*}, c+ B^{*}]$ with distribution: 

\begin{equation*} 
\mathbb{P}_{G}[s \geq p | s > c]= I_{\eta^{*}(\alpha^{*})} \left(\frac{r^{*}}{p-c} \right) ~~~ c + r^{*} < p  \leq c +B^{*}.
\end{equation*}

\end{theorem}

\noindent Remarkably, the probability of the lowest signal above $c$ is $\alpha \cdot(1-e^{-\eta^{*}(\alpha)})=1-e^{-\eta}$, which is independent of $c$. An additional atom might arise with costs, as in the case of $\eta=0$. However, the shape of the distribution $G$ conditional on trade also depends on costs through the disagreement level $\eta^{*}(\alpha^{*})$. With disagreement, costs further influence the shape by increasing the effective amount of disagreement conditional on trade, from $\eta$ to $\eta^{*}(\alpha)$. Thus, the probability of trade also determines the disagreement the buyer must accommodate conditional on trade. Since $\eta^{*}(\alpha)$ decreases in $\alpha$, a larger trading branch reduces this conditional disagreement.

\section{Additional Discussion and Results} \label{sec:extensions}

As is common with buyer-optimal learning, there is typically no closed form solution for the optimal price itself, although this can be calculated numerically. However, I can obtain an explicit characterization of how much the buyer is hurt by disagreement in the small $\eta$ limit: 

\begin{corollary} 
Suppose $c=0$. Let $r_{\eta}^{*}$ denote the seller's profit with disagreement $\eta$. Then there exists a positive and finite constant $k^{*}$ such that: 

\begin{equation*} 
r_{\eta}^{*} - r_{0}^{*} = k^{*} \sqrt{\eta} + o(\sqrt{\eta}). 
\end{equation*}
\end{corollary}

\noindent Thus, on the one hand, the impact of disagreement on profit is large relative to the size of disagreement itself, measured by $\eta$. On the other hand, despite the discontinuity exhibited by the equal-profit distribution (\ref{eq:RSInfo}), ultimately a small modification in the induced demand curve ensures only a small amount of surplus is lost if disagreement is small. 

Next, I discuss the role of tiebreaking in the proof, which ensures the seller charges a price which gives the buyer all residual surplus. For any other price $p$ to deliver profit $r$ under some $H_{p}$, $H_{p}$ must satisfy $\mathbb{P}_{H}[s \geq p] \geq r/p$, while $\mathbb{P}_{G_{r}}[s \geq p]=I_{\eta}(r/p) < r/p$. By Step One: 

\begin{equation*} 
D_{KL}(H_{p}||G) \geq d_{KL}(r/p|| I_{\eta}(r/p)) = \eta.
\end{equation*}

\noindent But since $H_{p}$ is admissible, $\eta$ is also an upper bound on the KL-divergence so this must hold with equality. In other words, the solution to the problem derived in Step One---which depends on $p$---is precisely the one which makes the seller indifferent between $p$ and $r$. 

But with this in mind, consider the set of \emph{feasible payoff profiles} over all $G$, again relaxing the assumption that ties are broken in favor of the buyer. Specifically, I consider payoff pairs supported by an admissible seller belief that attains the buyer's worst-case surplus given some tiebreaking rule. It is immediate that changing the tiebreaking rule cannot make the buyer \emph{better} off. More generally, define $BR(H) = \argmax_{p} p \cdot \mathbb{P}_{H}[s \geq p]$, and for any $t : \Delta([0,1]) \rightarrow [0,1]$ satisfying $t(H) \in BR(H)$, define

\begin{equation*} 
\mathcal{H}_{\eta}(G,t) = \argmin_{H : D_{KL}(H||G) \leq \eta } \int_{t(H)}^{1}(s-t(H))dG(s). 
\end{equation*}

\noindent and 
\begin{equation*}  \mathcal{P}_{\eta}=\{(U, \pi) :  \exists G, t, H \in \mathcal{H}_{\eta}(G,t), \text{ with } U=\int_{t(H)}^{1} (s-t(H))dG(s), \pi=t(H)\mathbb{P}_{H}[s \geq t(H)] \}. \end{equation*} The previous paragraph implies that the set of payoff pairs satisfying $\pi \geq r^{*}, U \geq 0$ and $U+\pi \leq \mu$ are contained in $\mathcal{P}_{\eta}$: To see this, take $G_{\pi}$ and use the $H_{\pi}$ delivering $U$, breaking ties against the buyer for that $H_{\pi}$ but in favor for all others. Since all prices in the support are weakly optimal for some $H_{\pi}$, this construction gives a triangle of feasible payoff profiles.

However, while these inequalities characterize feasible outcomes exactly when $\eta =0$, the same no longer holds with disagreement: 

\begin{proposition} 
Suppose $c=0$. For all $\eta > 0$, $\mathcal{P}_{\eta} \not\subset \mathcal{P}_{0}$ and $\mathcal{P}_{0} \not\subset \mathcal{P}_{\eta}$. 
\end{proposition}

\noindent That $\mathcal{P}_{0} \not\subset \mathcal{P}_{\eta}$ is intuitive as the buyer strictly does worse with disagreement. For the other noninclusion, I perturb the \cite{RS} information structure and seller belief so that the highest price becomes strictly optimal for the seller. This delivers zero buyer surplus while also lowering the seller's perceived probability of trade. I defer the formal description to the appendix.\footnote{Implicitly $\mathcal{P}_{\eta}$ uses $H$ to evaluate the seller's payoff. In the Supplemental Material, I modify the example to show the conclusion of the proposition remains valid if instead seller payoffs are evaluated with respect to $G$.} 

Lastly, I discuss the possibility of generalizing beyond KL divergences to parameterize disagreement. A natural generalization is to $f$-divergences, where for a convex function $f$ and any $H$ absolutely continuous with respect to $G$:
 \begin{equation*}
D_{f}(H|| G) =  \mathbb{E}_G \left [f\left (\frac{dH}{dG}\right) \right ],
\end{equation*}

\noindent where $\frac{dH}{dG}$ is the Radon-Nikodym derivative of $H$ with respect to $G$ (i.e., the function $m(s)$ satisfying $H(ds)=m(s)G(ds)$). The main reason why KL-divergence was useful in particular was due to the ability to invoke the KL-chain rule, used in Step One and in the Proof of Lemma  \ref{lem:conditioning}. But indeed, the same conclusion in Step One holds for $f$-divergences with a slightly more involved argument which I relegate to the Supplemental Appendix. The intuition is that since all that matters is the relative likelihood of the event that the buyer buys, the same bound can be achieved constructively by only adjusting the likelihood ratio between these two events. With this alternative, $I_{\tilde{\eta}}(x)$ would be defined using $d_f(x|| I_{\tilde{\eta}}(x))=\tilde{\eta}$ where: 
\begin{equation*}
d_f(x|| y)= y f(x/y) + (1-y)f((1-x)/(1-y)).
\end{equation*}
\noindent While $f$ would require some additional regularity conditions to ensure  $I_{\tilde{\eta}}(x)$ inherits the same monotonicity properties as with KL-divergence,\footnote{For instance, $f(1)=0, f''>0, f(0) < \infty$, $\lim_{t \rightarrow \infty} \frac{f(t)}{t} = \infty$, with $f \in C^{2}(0,\infty)$ continuous at 0. I verify that these conditions ensure the required conditions on $I_{\tilde{\eta}}$ are satisfied.} the same argument goes through. This transformation thus generalizes the theorem for the $c=0$ case.

The chain rule application in Lemma \ref{lem:conditioning}, used in the $c > 0$ case---which applies it while varying the probability of purchase---is harder to dispense with since the precise formula plays a role. To be sure, this step could be modified using analogous identities for other $f$-divergences. But these might change how $\eta^*(\alpha)$ varies in $\alpha$ and hence require the theorem to be restated. 

\newpage

\section{Detailed Proofs} \label{sec:proofs}

\begin{proof}[Proof of Lemma \ref{lem:bound}] 
Suppose the buyer attempts to hold the seller to profit at most $r$ when profit is evaluated according to $H$.  Since the seller chooses $p$ optimally, this requires $ \mathbb{P}_{H}[s \geq p] \leq r/p,$ for all $p > r$ and admissible $H$; for $p \leq r$ the restriction is automatic. Step One showed that given $G$ and $p$, the minimum KL ``cost of changing the probability of trade'' from $\mathbb{P}_{G}[s \geq p]$ to $\mathbb{P}_{H}[s \geq p]$ is $d_{KL}(\mathbb{P}_{H}[s \geq p] ||  \mathbb{P}_{G}[s \geq p])$; such a change is admissible whenever this cost is at most $\eta$.   Now, if $\mathbb{P}_{G}[s \geq p] > r/p$, then the seller achieves profit greater than $r$ even when $H=G$. Otherwise, since $d_{KL}(x|| y)$ is decreasing in $y$ for $y\leq x$, the probability $r/p$ can be reached by some $H$ if and only if $\mathbb{P}_{G}[s \geq p]\geq I_{\eta}(r/p)$. To see this, note that at equality the definition of $I_{\eta}(r/p)$ implies that making trade occur with probability $r/p$ requires spending the entire entropy budget; if the inequality is strict the seller can assign even more probability to trade. The Lemma follows. 
\end{proof}

\begin{proof}[Proof of Lemma \ref{lem:BP}]
Let $B$ solve:

\begin{equation*}
r+\int_{r}^{B}I_{\eta}(r/p)dp=\mu.
\end{equation*}

\noindent Some $B \in (r,\overline{v}]$ exists, since the left hand side is continuous and strictly increasing in $B$ on this range, less than $\mu$ at $B=r$ if the seller does not get full surplus, and at least as large as $\mu$ at $B=\bar{v}$ by (\ref{mainsuffeq}). When $r=\mu$ take $B=\mu$.

It follows that $G_{r}$ is Bayes plausible, since for $x \in [r,B]$:

\begin{equation*}
 \int_{\underline{v}}^{x} G_{r}(s)ds= \int_{r}^{x} [1- I_{\eta}(r/s)]ds=x-r-\int_{r}^{x} I_{\eta}(r/p)dp \leq \int_{\underline{v}}^{x} F(s)ds,
\end{equation*}

\noindent where the inequality uses the condition $D_{\eta}(r,x)-x \geq 0$. Furthermore, the fact that the mean of $G_{r}$ is $\mu$ implies the same inequality holds for $x > B$, and the inequality holds trivially since CDFs are nonnegative when $x < r$. Hence Bayes plausibility holds. Whenever $p > r$, $\mathbb{P}_{G_{r}}[s \geq p] \leq I_{\eta}(r/p)$ by construction, so Lemma \ref{lem:bound} implies $p \mathbb{P}_{H}[s \geq p] \leq r$ for every admissible $H$. For $p \leq r$ the bound is automatic, while at $p=r$ absolute continuity and $\mathbb{P}_{G_{r}}[s \geq r]=1$ imply $\mathbb{P}_{H}[s \geq r]=1$ for every admissible $H$, so that charging $r$ yields exactly the desired profit level $r$. 
\end{proof} 

\begin{proof}[Proof of the Price inducibility Claim]  I show that if \begin{equation*}
\mathbb{P}_{G}[s \geq p] \geq I_{\eta}\left(\frac{r^{*}}{p-p^{\dagger}+r^{*}}\right)
\end{equation*} for some $p > p^{\dagger}$, there is an $H$ under which $p$ strictly dominates every price no greater than $p^{\dagger}$. There are two cases. If \begin{equation*} 
\mathbb{P}_{G}[s \geq p] \geq \frac{r^{*}}{p-p^{\dagger}+r^{*}},\end{equation*}

\noindent then $H=G$ works. Indeed, if $q \leq p^{\dagger}$, then $q \mathbf{1}[s \geq q] \leq \min \{s, p^{\dagger}\}$ implies $q \mathbb{P}_{G}[s \geq q] \leq p^{\dagger} - \int_{0}^{p^{\dagger}} G(s)ds =r^{*}$, by the definition of $p^{\dagger}$. But $p$ achieves at least $p \frac{r^{*}}{p-p^{\dagger}+r^{*}}$, and hence strictly more than $r^{*}$ since $p^{\dagger} > r^{*}$.  Thus, the main case of interest is: 

\begin{equation*}
\frac{r^{*}}{p-p^{\dagger}+r^{*}}>\mathbb{P}_{G}[s \geq p] \geq I_{\eta}\left(\frac{r^{*}}{p-p^{\dagger}+r^{*}}\right)
\end{equation*}

\noindent Now, the definition of disagreement inversion implies that $d_{KL} \left( \frac{r^{*}}{p-p^{\dagger}+r^{*}} || I_{\eta}\left(\frac{r^{*}}{p-p^{\dagger}+r^{*}}\right) \right) = \eta$; and since $d_{KL}$ is decreasing in its second argument, Step One implies there is an admissible $H$ such that the probability of sale is $\frac{r^{*}}{p-p^{\dagger}+r^{*}} $. Again the seller's profit is at least $r^{*}$ for such an $H$, and furthermore, $H$ is a rescaling of $G$. That implies, for $s \leq p^{\dagger}$, 

\begin{equation*} 
H(s) = \frac{ 1- \frac{r^{*}}{p-p^{\dagger}+ r^{*}}}{ 1 - \mathbb{P}_{G}[s \geq p]}G(s). 
\end{equation*}
Integrating and using the identity that $\int_{0}^{p^{\dagger}} G(s)ds=p^{\dagger}-r^{*}$, 

\begin{equation*} 
\int_{0}^{p^{\dagger}} H(s)ds =\frac{ 1- \frac{r^{*}}{p-p^{\dagger}+ r^{*}}}{ 1 - \mathbb{P}_{G}[s \geq p]}(p^{\dagger}-r^{*}).
\end{equation*}

\noindent Yet the same argument as in the first case implies $q \mathbb{P}_{H}[s \geq q] \leq p^{\dagger} - \int_{0}^{p^{\dagger}} H(s)ds$. Subtracting the seller profit at $p$ from both sides and substituting in for the integral: 

\begin{align*} 
q \mathbb{P}_{H}[s \geq q] - p \mathbb{P}_{H}[s \geq p] &\leq p^{\dagger}-\frac{ 1- \frac{r^{*}}{p-p^{\dagger}+ r^{*}}}{ 1 - \mathbb{P}_{G}[s \geq p]}(p^{\dagger}-r^{*})-p \frac{r^{*}}{p-p^{\dagger}+r^{*}} \\ & =- \left( 1- \frac{r^{*}}{p-p^{\dagger}+ r^{*}} \right)(p^{\dagger} - r^{*}) \frac{ \mathbb{P}_{G}[s \geq p]}{1- \mathbb{P}_{G}[s \geq p]} < 0. 
\end{align*}

\noindent Hence, $p$ achieves more profit given $H$ than any $q \leq p^{\dagger}$. So every seller best response under this $H$ exceeds $p^{\dagger}$, and thus gives the buyer less than $\mu-r^{*}$. 
\end{proof}

\begin{proof}[Proof of Lemma \ref{lem:conditioning}]

If $\alpha=0$, buyer surplus is 0 because $s \leq c$ almost surely, while if $\alpha=1$ absolute continuity implies $\mathbb{P}_{H}[s > c]=1$ for every admissible $H$; hence $D_{KL}(H||G) = D_{KL}(H||G;s > c)$, and the claimed reduction holds with $\eta^{*}(1)=\eta$. For the remaining cases, take $0 < \alpha < 1$. First suppose $\alpha > 1-e^{-\eta}$. Let $\hat{\alpha}= \mathbb{P}_{H}[s > c]$; $\hat{\alpha} > 0$ since $d_{KL}(0||\alpha) = - \log(1- \alpha) > \eta$. But since $\hat{\alpha} > 0$, some price strictly above $c$ yields strictly positive profit, so every optimal price satisfies $p > c$.  Seller profit is $\hat{\alpha}(p-c) \mathbb{P}_{H}[ s \geq p \mid s > c]$, so the profit-maximizing $p$ is independent of $\hat{\alpha}$ and $H$ conditional on $s \leq c$. Applying the KL chain rule: 

\begin{equation*} 
D_{KL}(H || G) =  d_{KL}(\hat{\alpha} ||  \alpha) + \hat{\alpha}D_{KL}(H || G; s > c)+(1-\hat{\alpha})D_{KL}(H || G; s \leq c).
\end{equation*}

\noindent For any fixed choice of $H$ conditional on $s > c$, the least-costly implementation takes $H=G$ conditional on $s \leq c$ and minimizes the resulting KL cost over $\hat{\alpha}$, as this can only decrease $D_{KL}(H||G)$. Taking the FOC for $\hat{\alpha}$ yields: 

\begin{equation*} 
\min_{\hat{\alpha}} \overbrace{d_{KL}(\hat{\alpha}|| \alpha) + \hat{\alpha} D_{KL}(H || G ; s > c)}^{(*)} \Rightarrow \frac{\hat{\alpha}}{1- \hat{\alpha}}= e^{- D_{KL}(H || G ; s > c)} \frac{\alpha}{1- \alpha}.
\end{equation*}

\noindent Substituting in for $\hat{\alpha}$ fixing $D_{KL}(H || G ; s > c)$ yields that $(*)$ is $-\log(1- \alpha+\alpha e^{-D_{KL}(H || G ; s > c)})$, an expression that is strictly increasing in $D_{KL}(H || G; s > c)$. Hence, for this choice of $H$ conditional on $s>c$, the minimum entropy cost is at most $\eta$ if and only if
$D_{KL}(H||G;s>c) \leq \eta^*(\alpha)$. Thus, the choice of $H$ conditional on $s>c$ coincides with the problem where $\alpha=1$, but with the entropy radius modified accordingly. 

On the other hand, if $\alpha \leq 1- e^{-\eta}$ then for any $H$, 

\begin{equation*} 
-\log(1-\alpha+ \alpha e^{-D_{KL}(H || G; s > c)}) \leq -\log(1- \alpha) \leq \eta,
\end{equation*}

\noindent for every $H$ absolutely continuous with respect to $G$. In particular, for any $q$ below the supremum of the support conditional on $s > c$, let Nature choose $H$ so that conditional on $s > c$, $H$ is equal to $G$ further conditioned on $s \geq q$; $D_{KL}(H||G; s > c) < \infty$ and hence there is some such admissible $H$. Then the seller's optimal price is $q$ or greater; setting $q \rightarrow \sup \text{Supp} G$ drives buyer surplus to zero. 

Thus, for $\alpha > 1-e^{-\eta}$, fixing the distribution of $v$ induced by the original information structure conditional on $s > c$, the feasible conditional seller beliefs are exactly those in an $\eta^{*}(\alpha)$-entropy ball. By the remark from the main text, after shifting values and prices by $c$, the conditional problem is exactly the $c=0$ problem with $\eta=\eta^{*}(\alpha)$; since only signals above $c$ can trade at an optimal price $p > c$, unconditional buyer surplus and seller profit evaluated under $G$ are $\alpha$ times the conditional values.  \end{proof}

\begin{proof}[Completing the Proof of the $c > 0$ Theorem] First, I show that buyer surplus is positive if and only if $c < \bar{c}$ By Lemma \ref{lem:conditioning}, positive buyer surplus requires $\alpha > 1-e^{-\eta}$. By Lemma 1 in \cite{kolotilin}, maximizing the probability that $s \geq c$ is achieved by an information structure that tells the buyer if $v \geq \bar{v}$ where $\mathbb{E}[v \mid v \geq \bar{v}]=c$, with randomization at the cutoff if necessary; this probability can be arbitrarily well approximated by the probability that $s > c$ for some information structure. Hence a probability $\alpha$ can be generated only if 

\[
\frac{1}{\alpha} \int_{1-\alpha}^{1} F^{-1}(u)du > c.
\] If instead $c < \mu$, this condition holds immediately. Hence $c < \bar{c}$ is necessary.

For sufficiency, if $c < \mu$, applying the construction to $v-c$ delivers positive surplus immediately since the mean is positive. If instead $\mu \leq c < \bar{c}$, take $\alpha > 1-e^{-\eta}$ sufficiently close to $1-e^{-\eta}$ such that $\frac{1}{\alpha} \int_{1-\alpha}^{1} F^{-1}(u)du > c$. There exists an information structure generating a high signal with probability $\alpha$ and this conditional mean by revealing whether $v$ is in the upper-$\alpha$ quantile. Conditional on this event, $v-c$ has strictly positive mean; hence again the same construction delivers strictly positive surplus.

Lastly, I derive $\Gamma_{\eta}$. Define $z=s-c$ to be the posterior expectation net of costs, so $\mathbb{P}_G[z>0]=\alpha$ and $\mathbb{E}_G[z]=\mu-c$. I start with the following inequality
\begin{equation*}
z \leq \min\{x,r\}+ \max\{z-x,0\}+\int_{r}^{\max\{x,r\}} \mathbf{1}[z \geq t] dt,
\end{equation*}

\noindent which coincides with (\ref{eq:prelimineq}) when $x > r$ and reduces to $z \leq x +\max\{z-x,0\}$ otherwise. Taking expectations conditional on $z > 0$ and multiplying by $\alpha$, 

\begin{equation*}
\alpha \mathbb{E}_{G}[z \mid z > 0] \leq \alpha \min\{x,r\}+\alpha \mathbb{E}_{G}\left[\max\{z-x,0\} \mid z > 0 \right]+\alpha \int_{r}^{\max\{x,r\}}  \mathbb{P}_{G}[z \geq t \mid z > 0] dt.
\end{equation*}

Using $\mathbb{E}_{G}[\max\{z-x,0\}] = \mu - c- x + \int_{0}^{c+x}G(s)ds$ and Bayes plausibility: 

\begin{equation*} 
\alpha \mathbb{E}_{G}[\max\{z-x,0\} \mid z > 0] \leq \mathbb{E}_{G}[\max\{z-x,0\}] \leq \mu - c- x + \int_{0}^{c+x}F(s)ds.
\end{equation*}

\noindent Thus, with the modified Lemma \ref{lem:bound}, I obtain the following upper bound on $\alpha \mathbb{E}_{G}[z \mid z > 0]$: 

\begin{equation*} 
\alpha \mathbb{E}_{G}[z \mid z > 0] \leq \overbrace{\min_{x \in [-c,1-c]}  \alpha \min\{x,r\}+\mu - c- x + \int_{0}^{c+x}F(s)ds+\alpha \int_{r}^{\max\{x,r\}} I_{\eta^{*}(\alpha)}(r/t)dt.}^{\Gamma_{\eta}(\alpha,r)}
\end{equation*}

Now I construct the signal achieving the upper bound, taking $\alpha > 1-e^{-\eta}$ and $r > 0$ with $\Gamma_{\eta}(\alpha,r) \geq \max\{\mu-c, \alpha r\}$ (which is necessary as argued in the main text). If $\alpha=1$, the construction coincides with the Section \ref{sec:c0} problem. For $\alpha <1$, define the pooled shifted signal: 

\begin{equation*} 
\lambda = \frac{\mu-c- \Gamma_{\eta}(\alpha,r)}{1-\alpha} \leq 0.
\end{equation*}

\noindent Evaluating $\Gamma_{\eta}$ at $x=-c$ gives $\Gamma_{\eta}(\alpha,r) \leq \mu-\alpha c$, so $\lambda \in [-c,0]$ and the pooled posterior $c+\lambda$ lies in $[0,c]$. Also, $r \leq 1-c$ must hold, since if $r > 1-c$ then evaluating $\Gamma_{\eta}(\alpha,r)$ at $x=1-c$ yields: 
\begin{equation*} 
\Gamma_{\eta}(\alpha,r) \leq \alpha(1-c) < \alpha r,
\end{equation*}

\noindent contradicting $\Gamma_{\eta}(\alpha,r) \geq \alpha r$. Evaluating $\Gamma$ at $x=1-c$ thus yields: 

\begin{equation*} 
\Gamma_{\eta}(\alpha,r)  \leq \alpha r + \alpha \int_{r}^{1-c} I_{\eta^{*}(\alpha)}(r/t)dt.
\end{equation*}

Since $\Gamma_{\eta}(\alpha,r) \geq \alpha r$, continuity implies there exists some $B \in [r,1-c]$ such that: 

\begin{equation*} 
\Gamma_{\eta}(\alpha,r)  = \alpha r + \alpha \int_{r}^{B} I_{\eta^{*}(\alpha)}(r/t)dt. 
\end{equation*}

\noindent Consider the information structure where $s=c+\lambda$ is realized with probability $1-\alpha$ and otherwise uses the distribution in the theorem with parameters $(\alpha,r,B)$; write $G$ for the CDF of $s$ and $G_{\alpha,r}(z)=G(c+z)$ for the CDF of the shifted signal, so that $\int_{-c}^{x} G_{\alpha,r}(z)dz=\int_{0}^{c+x} G(s)ds$. It remains to check Bayes plausibility. For $x \in [ \lambda, r]$, 

\begin{equation*} 
\int_{-c}^{x} G_{\alpha,r}(z) dz= (1- \alpha)(x-\lambda) = \Gamma_{\eta}(\alpha,r)-(\mu-c)+(1-\alpha)x. 
\end{equation*}

\noindent Using the expression for $\Gamma_{\eta}(\alpha,r)$ at $x < r$,

\begin{equation*} 
\Gamma_{\eta}(\alpha,r) \leq \alpha x + \mu -c-x+\int_{0}^{c+x}F(s)ds.
\end{equation*}

\noindent Putting the previous two equations together, $\int_{0}^{c+x}G(s)ds \leq \int_{0}^{c+x} F(s)ds$. For $x \in [r,B]$, the same calculation, evaluated at $x \geq r$ yields: 

\begin{multline*} 
\int_{-c}^{x} G_{\alpha,r}(z) dz= \Gamma_{\eta}(\alpha,r)-(\mu-c) +(1-\alpha)x+ \alpha \int_{r}^{x} 1-I_{\eta^{*}(\alpha)}(r/t)dt \\  \Gamma_{\eta}(\alpha,r) \leq \alpha r + \mu -c-x+\int_{0}^{c+x}F(s)ds +\alpha \int_{r}^{x} I_{\eta^{*}(\alpha)}(r/t)dt.
\end{multline*}

\noindent Hence again $\int_{0}^{c+x} G(s)ds \leq \int_{0}^{c+x} F(s)ds$. For $x < \lambda$, $\int_{0}^{c+x} G(s)ds=0$, so this inequality holds, and for $x > B$ it follows from exactly the same argument as the proof of Lemma \ref{lem:BP}. Thus the construction is Bayes plausible; since the seller charges $c+r$, buyer surplus is $\Gamma_{\eta}(\alpha,r) - \alpha r$. 

Lastly, note that since $I_{\eta^{*}(\alpha)}(r/t) \leq I_{\eta^{*}(\alpha)}(1) = e^{-\eta^{*}(\alpha)}$: 

\begin{equation*}
\Gamma_{\eta}(\alpha,r) - \alpha r \leq \alpha \int_{r}^{1-c} I_{\eta^{*}(\alpha)}(r/t) dt \leq \alpha e^{-\eta^{*}(\alpha)},
\end{equation*}

\noindent which converges to 0 as $\alpha \rightarrow 1-e^{-\eta}$. Thus, whenever the optimal value is positive, the optimal value is bounded away from this boundary, which by the above is possible whenever $c < \bar{c}$. Together with $0 \leq r \leq 1-c$ and continuity, the maximum is attained. 
\end{proof}

\begin{proof}[Proof of the Corollary] 
I first show $I_{\eta}(x) =x -\sqrt{2x(1-x)} \sqrt{\eta} +o(\sqrt{\eta})$. This uses that $\frac{\partial}{\partial y} d_{KL}(x||y) = \frac{y-x}{y(1-y)}$. Using the Taylor expansion of $d_{KL}(x||x-\Delta)$ around $\Delta=0$, 

\begin{equation*} 
d_{KL}(x || x- \Delta) = \frac{\Delta^{2}}{2x(1-x)} +o(\Delta^{2}).
\end{equation*}

\noindent Setting $\Delta=x-I_{\eta}(x)$, so that $d_{KL}(x||x-\Delta)=\eta$, yields the claimed approximation.

Now, using the bound that $d_{KL}(x||y) \geq 2(x-y)^{2}$ (i.e., Pinsker's inequality for Bernoulli random variables), the definition of $I_{\eta}(q)$ implies: 

\begin{equation*} 
0 \leq \frac{q - I_{\eta}(q)}{ \sqrt{\eta}} \leq \frac{1}{\sqrt{2}}. 
\end{equation*}

\noindent Dominated convergence and the previous approximation imply: 

\begin{equation*} 
\frac{1}{\sqrt{\eta}} \int_{r}^{x} \frac{r}{p} - I_{\eta}(r/p) dp \rightarrow \int_{r}^{x}\sqrt{ 2\frac{r}{p}\left(1 - \frac{r}{p} \right)} dp
\end{equation*}

\noindent Letting $\Phi_{\eta}(r,x)=D_{\eta}(r,x)-x$, using this approximation for $I_{\eta}(r/p)$ I obtain: 

\begin{equation*} 
\Phi_{\eta}(r,x)=\Phi_{0}(r,x) - \sqrt{\eta} \int_{r}^{x}\sqrt{ 2\frac{r}{p}\left(1 - \frac{r}{p} \right)} dp  +o(\sqrt{\eta}).  
\end{equation*}

Let $V_{\eta}(k)=\min_{x}\Phi_{\eta}(r_{0}+k\sqrt{\eta},x)$ and $X_{0}$ denote the set of minimizers at $\eta=0$. Note that $x > r_{0}$ at any minimizer, since:

\begin{equation*} 
\Phi_{0}(r_{0},r_{0})= \int_{0}^{r_{0}} F(s) ds > 0, 
\end{equation*}

\noindent while $V_{0}(k)=0$. Since $r_{0} > 0$, the aforementioned convergence is uniform for $r$ in a neighborhood of $r_{0}$ and $x \in [r,1]$; hence, the $o(\sqrt{\eta})$ remainder can be taken to be uniform over the relevant minimization. Thus, applying the same approximation for $\Phi_{\eta}(r,x)$ along $r=r_{0} + k \sqrt{\eta}$, the envelope theorem \cite[Corollary 4]{milgrom2002envelope} implies:

\begin{equation*} 
\lim_{\eta \rightarrow 0} \frac{V_{\eta}(k) - V_{0}(k)}{\sqrt{\eta}}= \min_{x \in X_{0}} k \log(x/r_{0}) -\int_{r_{0}}^{x}\sqrt{ 2\frac{r_{0}}{p}\left(1 - \frac{r_{0}}{p} \right)} dp   .
\end{equation*}

\noindent Set $k^{*}= \max_{x \in X_{0}} \frac{\int_{r_{0}}^{x}\sqrt{ 2\frac{r_{0}}{p}\left(1 - \frac{r_{0}}{p} \right)} dp}{\log(x/r_{0})}$, which is finite by interiority of any minimizing $x \in X_{0}$. Now, note that $\min_{x \geq r} \Phi_{\eta}(r,x)$ is nondecreasing in $r$; indeed, for each $x$: 
\begin{equation*} 
D_{\eta}(r_{2},x)-D_{\eta}(r_{1},x) = (r_{2} - r_{1}) - \int_{r_{1}}^{r_{2}} I_{\eta}(r_{1}/p) dp + \int_{r_{2}}^{x} [I_{\eta}(r_{2}/p) - I_{\eta}(r_{1}/p)] dp,
\end{equation*} 

\noindent which is nonnegative because $I_{\eta} \leq 1$ and $I_{\eta}$ is increasing. Hence, as $r$ increases, the objective in the minimization problem increases for all $x$, while the set of feasible $x$ shrinks. Now, $k < k^{*}$ implies for all $\eta$ sufficiently small, $V_{\eta}(k) < 0$; likewise, $k > k^{*}$ implies for all $\eta$ sufficiently small, $V_{\eta}(k) > 0$. Since $r_{\eta}^{*}$ is the threshold value for $r$ at which $V_{\eta}(k)$ becomes nonnegative, it follows that for all $\varepsilon > 0$, 

\begin{equation*} 
r_{0} + (k^{*} - \varepsilon) \sqrt{\eta} < r_{\eta} < r_{0} + (k^{*} + \varepsilon) \sqrt{\eta},
\end{equation*}

\noindent whenever $\eta$ is sufficiently small. Since $\frac{r_{\eta} - r_{0}}{\sqrt{\eta}} \rightarrow k^{*}$ for some fixed $k^{*}$, the corollary holds.\end{proof}

\begin{proof}[Proof of the Proposition]  For $\mathcal{P}_{0} \not\subset \mathcal{P}_{\eta}$, recall that $x^{*} > r^{*}$ if $r^{*} < \mu$. Let $r_{\eta}$ denote the value of $r^{*}$ as a function of $\eta$ and let $x_{0}$ denote any value of $x^{*}$ at $\eta=0$. Now, $I_{0}(r_{0}/p)=r_{0}/p> I_{\eta}(r_{0}/p)$ for any $p \in (r_{0}, x_{0}]$. So, $D_{\eta}(r_{0},x_{0}) - x_{0} < D_{0}(r_{0},x_{0})-x_{0}=0$, and hence $r_{\eta} > r_{0}$; by the theorem, $\mu-r_{\eta} < \mu - r_{0}$, so the buyer-optimal vertex of the $\eta=0$ case cannot belong to $\mathcal{P}_{\eta}$. 

For $\mathcal{P}_{\eta} \not\subset \mathcal{P}_{0}$, letting $G_{RS}$ denote the Roesler-Szentes information structure, define: 

\begin{equation*} 
G_{\varepsilon} =(1- \varepsilon) G_{RS} + \varepsilon (\lambda \delta_{\mathbb{E}[v \mid v \leq x]} + (1-\lambda) \delta_{\mathbb{E}[v \mid v > x]}),
\end{equation*}

\noindent for some $x$ such that $\mathbb{E}[v \mid v \leq x] < r_{0} < \mu < \mathbb{E}[v \mid v > x] <B$; note this holds whenever $x$ is sufficiently small and $\lambda= \mathbb{P}[v \leq x]$, and that this distribution is Bayes plausible with maximum posterior equal to $B$. Now for $\kappa, \tau$ sufficiently small, let $H=(1- \kappa - \tau) G_{RS} + \kappa \delta_{\mathbb{E}[v \mid v \leq x]} + \tau \delta_{B}$.

Seller profit following price $p$ is $(1-\kappa - \tau)r_{0} + \tau p$ whenever $p \in [r_{0},B]$; any higher price delivers zero profit, and any price less than $r_{0}$ delivers lower profit given $x$ small. So, profit is maximized at $p=B$, and profit at this price is less than $r_{0}$ whenever:

\begin{equation*} 
\frac{\tau}{\kappa + \tau} < \frac{r_{0}}{B}.
\end{equation*}

\noindent But because the seller chooses $p=B$, the buyer obtains 0 surplus. Taking $\kappa = \kappa^{*} \varepsilon$ and $\tau=\tau^{*} \varepsilon$ to satisfy this inequality, a direct calculation shows $D_{KL}(H||G_{\varepsilon}) \rightarrow 0$ as $\varepsilon \rightarrow 0$. Since buyer surplus is minimized (being 0), this payoff pair is feasible; thus, for every $\eta > 0$, some point in this set achieves revenue lower than $r_{0}$ and buyer surplus 0, which is not achievable in $\mathcal{P}_{0}$. \end{proof}

\newpage 

\section{Online Supplement}

\begin{example}  \label{ex:nocanon}
This example illustrates why the replacement strategy of \cite{RS} cannot be used to prove the Theorem. Let $F=U[0,1]$ and $\eta=.1$. Let the set of signals be $\{.25, .45, .7, .9\}$, where $G(\{.25\})=.4, G(\{.45\})=.2, G(\{.7\})=.25, G(\{.9\})=.15$. I claim that a worst-case belief is $H=G$. First, calculating the seller revenue at each possible price (which is in the support of the buyer's posterior distribution), the seller obtains $.25$ by charging $.25$, $.27$ by charging $.45$, $.28$ by charging $.7$, and $.135$ by charging $.9$. Thus, it follows that $.7$ is uniquely optimal when $H=G$, and the buyer obtains payoff $.15(.9-.7)=.03$. 

Crucially, since $\eta=.1$, it is not possible for Nature to make $.9$ an optimal price. To do so, since the seller obtains $.25$ by charging $.25$, any such $H$ would require the seller to achieve at least this much at $p=.9$ and no more at any other price. To make $.9$ optimal, its revenue must weakly exceed the revenue at each of the three lower prices. At the KL-minimizing adjustment, all three of these revenue constraints bind, yielding $H(\{.25\})=4/9, H(\{.45\})=25/126, H(\{.7\})=5/63, H(\{.9\})=5/18$. But I compute $D_{KL}(H||G) \approx .125 > .1$; hence no admissible $H$ can make a price of $.9$ optimal. Since $H=G$ makes $.7$ optimal, it is therefore a worst-case belief. 

Since $H=G$ in this example, seller surplus according to either player is $.7(.4)=.28$. Attempting to follow the proof strategy of \cite{RS} would then require us to find a full-trade information structure achieving this surplus.\footnote{The property that $H=G$ is helpful since it is not important to take a stand whether this efficiency change should be evaluated according to the seller's belief or the buyer's belief.} Of course, their replacement works when $\eta=0$. That this is not possible when $\eta=0.1$ follows from the proof of the Theorem, which derives the necessary condition that (\ref{mainsuffeq}) holds for each $x \geq r$. So, a seller profit of no more than $.28$ is achievable under full trade only if for all $x \geq .28$: 

\begin{equation*} 
x \leq .28 + \frac{x^{2}}{2} + \int_{.28}^{x} I_{\eta}(.28/p) dp.
\end{equation*}

\noindent Numerically evaluating the integral at $x=.7$ yields $\int_{.28}^{x} I_{\eta}(.28/p) dp \leq .17$, so the right hand side is less than $.28+ \frac{.7^2}{2} + .17 =.695$. Hence, the inequality is violated and it is not possible to hold the seller to profit at most $.28$. \medskip
\end{example}

\noindent \textbf{The Proposition Evaluating Seller Profit Using Buyer Beliefs} The same construction from the proof of the Proposition identifies a payoff profile with disagreement that cannot arise without disagreement, even when evaluating seller profit with respect to $G$ instead of $H$. Let: 

\begin{equation*} 
G_{\varepsilon} = (1- \varepsilon)G^{RS} + \varepsilon \delta_{\mu} ~~~ \text{ and } ~~~   H_{t} = (1- t) G^{RS} + t \delta_{B}.
\end{equation*}

\noindent By construction, $B$ is uniquely optimal under $H_{t}$ and so buyer surplus is again 0. Seller surplus evaluated according to the buyer's distribution is $(1- \varepsilon) r_{0}$. Again, $D_{KL}(H_{t}||G_{\varepsilon}) \rightarrow 0$ as $t, \varepsilon \rightarrow 0$.

\medskip
\medskip

\noindent \textbf{Step One with $f$-divergences} Defining $m(s) = \frac{dH}{dG}(s)$, $\mathbb{P}_{H}[s \geq p]=\int_{p}^{1} m(s) dG(s)$, and hence also that $\mathbb{E}_{G}[m(s) \mid s \geq p] = \frac{1}{\mathbb{P}_{G}[s \geq p]} \int_{p}^{1} m(s)dG(s) =\frac{\mathbb{P}_{H}[s \geq p]}{\mathbb{P}_{G}[s \geq p]}  $. Similarly, $\mathbb{E}_{G}[m(s) \mid s < p]= \frac{\mathbb{P}_{H}[s < p]}{\mathbb{P}_{G}[s < p]}  $. Jensen's inequality since $f$ is convex then yields: 

\begin{equation*} 
\mathbb{E}_{G} \left[ f(m(s)) \mid s \geq p \right] \geq f\left(\frac{\mathbb{P}_{H}[s \geq p]}{\mathbb{P}_{G}[s \geq p]} \right) ~~ \text{ and } ~~ \mathbb{E}_{G} \left[ f(m(s)) \mid s < p \right] \geq f\left( \frac{\mathbb{P}_{H}[s < p]}{\mathbb{P}_{G}[s < p]}  \right).
\end{equation*}

\noindent Using these inequalities, I obtain the following: 

\begin{equation*} 
\mathbb{E}_{G}[f(m(s))] \geq \mathbb{P}_{G}[s \geq p]f\left(\frac{\mathbb{P}_{H}[s \geq p]}{\mathbb{P}_{G}[s \geq p]} \right) + \mathbb{P}_{G}[s < p]f\left( \frac{\mathbb{P}_{H}[s < p]}{\mathbb{P}_{G}[s < p]}  \right).
\end{equation*}

Hence I obtain the same lower bound as in Step One for KL divergence, and this bound is attained with equality by taking:
\[
m(s)=
\begin{cases}
\dfrac{\mathbb{P}_{H}[s\geq p]}{\mathbb{P}_{G}[s\geq p]}, & s\geq p,\\[2mm]
\dfrac{\mathbb{P}_{H}[s<p]}{\mathbb{P}_{G}[s<p]}, & s<p.
\end{cases}
\]
Thus, as before, it is without loss to change only the relative probability of signals above and below the price.

I also verify that under these conditions, the equation: 

\begin{equation*} 
\eta=I_{\eta}(x)f(x/I_{\eta}(x))+(1-I_{\eta}(x))f((1-x)/(1-I_{\eta}(x))),
\end{equation*}

\noindent has a unique solution for $I_{\eta}(x)$ less than or equal to $x$. At $I_{\eta}(x)=x$, the right hand side reduces to $f(1)=0$, while as $I_{\eta}(x) \rightarrow 0$, the right-hand side approaches $\infty$, since: 

\begin{equation*} 
I_{\eta}(x)f(x/I_{\eta}(x))=x \cdot \frac{f(x/I_{\eta}(x))}{x/I_{\eta}(x)},
\end{equation*}

\noindent and $f(t)/t \rightarrow \infty$ as $t \rightarrow \infty$. At $x=1$, the right hand side is finite since $f(0) < \infty$. Lastly, the derivative of the right hand side with respect to $x$ is $f'(x/y)-f'((1-x)/(1-y))$, which is positive for $y <x$ when $f$ is strictly convex. The derivative with respect to $y$ is $f(x/y)-(x/y)f'(x/y)-(f((1-x)/(1-y))-((1-x)/(1-y))f'((1-x)/(1-y)))$; but since: 

\begin{equation*} 
\frac{d}{dt}(f(t) - tf'(t))=-tf''(t), 
\end{equation*}

\noindent this derivative is negative when $f'' > 0$. It follows that there is a unique value of $I_{\eta}(x)$ for every $x$, and that $I_{\eta}(x)$ is also increasing. These properties are sufficient for the $c=0$ argument to generalize.

\newpage

\bibliographystyle{ecta}
\bibliography{SignalRobust}

\end{document}